\documentclass[10pt,a4paper]{article}

\usepackage[utf8]{inputenc}
\usepackage[T1]{fontenc}
\usepackage{microtype}
\usepackage{fourier}
\usepackage{tensor}
\usepackage{braket}

\usepackage{tikz}
\usetikzlibrary{arrows}
\usepackage[margin=10pt,font=small,labelfont=bf,labelsep=period]{caption}
\usepackage{amsmath}
\usepackage{amsthm}
\usepackage{amssymb}
\usepackage{mathscinet}
\numberwithin{equation}{section}

\usepackage[nottoc,notlot,notlof]{tocbibind}

\usepackage{aliascnt}
\usepackage[colorlinks,linkcolor=blue,citecolor=blue]{hyperref}

\theoremstyle{plain}

\newtheorem{theorem}{Theorem}[section]

\newaliascnt{lemma}{theorem}

\aliascntresetthe{lemma}

\newaliascnt{corollary}{theorem}

\aliascntresetthe{corollary}

\theoremstyle{definition}

\newaliascnt{definition}{theorem}

\aliascntresetthe{definition}

\newaliascnt{example}{theorem}
\newtheorem{example}[example]{Example}
\aliascntresetthe{example}

\newaliascnt{remark}{theorem}
\newtheorem{remark}[remark]{Remark}
\aliascntresetthe{remark}

\newcommand{\RR}{\mathbf{R}}

\renewcommand{\epsilon}{\varepsilon}

\newcommand{\abs}[1]{\left\lvert #1 \right\rvert}
\newcommand{\avg}[1]{\bigl\langle #1 \bigr\rangle}

\newcommand{\In}{\mathcal{L}_{\rho}}
\newcommand{\Inm}{\mathcal{L}_{m}}

\DeclareMathOperator{\sgn}{sgn}

\title{Isospectral flows for the inhomogeneous string density problem}
\author{Andrzej Z. G\'orski
\thanks{H. Niewodnicza\'nski Institute of Nuclear Physics, Polish Academy of Sciences,
Radzikowskiego 152, Krak\'ow, Poland; Andrzej.Gorski@ifj.edu.pl}
\and 
  Jacek Szmigielski
  \thanks{Department of Mathematics and Statistics, University of Saskatchewan, 106 Wiggins Road, Saskatoon, Saskatchewan, S7N 5E6, Canada; szmigiel@math.usask.ca}}
\date{\today}

\makeatletter
\hypersetup{%
  pdfauthor={Jacek Szmigielski},
  pdftitle={\@title},
}
\makeatother

\begin{document}
\maketitle
\begin{abstract} 
We derive isospectral flows of the mass density in the string boundary 
value problem corresponding to general boundary conditions.  In 
particular, we show that certain class of rational flows 
produces in a suitable limit all flows generated by polynomials in negative powers of the spectral parameter.  We illustrate the theory 
with concrete examples of isospectral flows of discrete 
mass densities which we prove to be Hamiltonian and for which we provide explicit solutions of equations of motion in terms of 
Stieltjes continued fractions and Hankel determinants.

\end{abstract} 

\section{Introduction} 

The  1-D wave equation 
$\tfrac{1}{c^2} u_{\tau \tau}-u_{xx}=0$
is a basic classical equation describing propagation of waves, in particular 
vibrations of a string.  The coefficient 
$c^2$ has the physical dimension of velocity squared and is usually 
assumed to be independent of position.  In general $\frac{1}{c^2}$ is 
proportional to the lineal mass density $\rho$ with the inverse of the 
proportionality constant being the string tension which 
for convenience we will set equal to $1$.  Thus 
we can write the inhomogeneous string wave equation as: 
\begin{equation} \label{eq:1d wave inh}
\rho(x) u_{\tau \tau}-u_{xx}=0.  
\end{equation} 
The normal modes $v$ are derived by substituting $u=v(x)\cos{\omega t}$ 
into \eqref{eq:1d wave inh}, resulting in 
\begin{equation} \label{eq:prelimstring}
v_{xx}=-\omega^2 \rho(x) v\stackrel{def}{=}-z \rho(x) v.   
\end{equation} 
Furthermore, one needs to specify the length of the string; in this paper we will choose $0\leq x\leq 1$ so that the string will have a unit length, as well as  we impose some boundary 
conditions reflecting how the string is tied at the endpoints.  The most typical boundary conditions take the 
form $v_x(0)-hv(0)=0$ and $v_x(1)+Hv(1)=0$ where both 
$h$ and $H$ are non-negative or infinity.  In the latter case, taking 
the left boundary condition as an example, the 
interpretation of $h=\infty$ is simply $v(0)=0$, in other words 
$h=\infty$ corresponds to 
the Dirichlet condition at $x=0$.  Thus in this paper 
the inhomogeneous string equation will mean the boundary value problem 
\begin{subequations} 
\begin{equation} \label{eq:string}
v_{xx}=-z \rho(x) v,  \qquad 0<x<1, 
\end{equation} 
\begin{equation} \label{eq:BCstring} 
v_x(0)-hv(0)=0, \qquad v_x(1)+Hv(1)=0, 
\end{equation} 
\end{subequations} 
where $0\leq h, H\leq \infty$. 
Physically those $z$ for which the boundary value problem has a solution 
represent the squares of frequencies and, clearly, they are expected to be 
positive, or zero.   The literature on the inhomogeneous string problem 
includes 
\cite{barcilon83, gantmacher-krein, kackrein, McKean, colville-gomez-sz, krein-string}
and recent contributions \cite{Amore, Amore-further} dealing with 
a perturbative approach to computations of densities close to the 
homogenous one.  Models leading, directly or indirectly, to 
inhomogeneous strings appear in many areas of science, from geophysics 
(see e.g. \cite{hss})  and fluid dynamics (see e.g. \cite{BSS-Stieltjes}) to particle physics (see e.g. \cite{DHoker}).  
\section{Isospectral deformations}
The question of finding \textit{isospectral}, that is leaving the 
spectrum invariant, deformations of the mass density $\rho$ appearing 
in the string equation \eqref{eq:string} was proposed in a sequel 
of interesting papers by P. Sabatier in \cite{Sabatier-Constants, Sabatier-Evolution, Sabatier-NEEDS}.  In simple terms one is looking for 
evolution equations on $\rho$ with respect to the deformation 
parameter $t$ (not to be confused with physical time) for which 
the spectrum of the boundary value problem \eqref{eq:string} remains the same.  Historically, this problem goes back at least to G. Borg \cite{Borg}
who was perhaps the first to systematically study
the question of determining the potential  in the Sturm-Liouville problem from the knowledge of eigenvalues; 
in its canonical form, called the Liouville normal form, the underlying equation is equivalent to the 1-D Schr\"odinger equation $$
-\psi_{xx}+u(x)\psi=E \psi, $$
and the main message of \cite{Borg} was that two spectra are needed to determine uniquely $u$.  
In other words, by knowing a single spectrum, we can only 
hope to determine a family of potentials, and that family will be 
isospectral relative to that chosen spectrum.  The situation for 
the string equation \eqref{eq:string} is analogous, even though 
the string boundary value problem cannot be put in the canonical form mentioned above for 
general mass densities $\rho$.  We  will refer to the 
resulting family as \textit{isospectral strings}.  An additional 
motivation for studying this problem comes from the theory of completely integrable nonlinear partial differential equations, especially the Camassa-Holm (CH) equation \cite{CH}
\begin{equation} \label{eq:CH} 
u_t-u_{xxt}+3uu_x-2u_xu_{xx}-uu_{xxx}=0,  
\end{equation} 
modelling wave propagation in shallow water.  The rationale for this 
connection comes from the result in \cite{BSS-Stieltjes} 
that the CH can be viewed as an isospectral deformation of the string equation with 
Dirichlet boundary conditions (see also reviews \cite{bss-string, EKT}).  
\subsection{Summary of the main results of \cite{colville-gomez-sz}}
This section contains a summary of the approach to isospectral strings 
taken in \cite{colville-gomez-sz} based on ideas originating in the theory of integrable systems.  
The starting point is to postulate the "time" deformation of \eqref{eq:string}
to be of the form
\begin{equation} \label{eq:vtdeform}
v_t=av+bv_x.  
\end{equation} 
Then a simple computation that amounts to checking $v_{xxt}=v_{txx}$ 
yields 
\begin{subequations}
\begin{align}
z\rho_t=\tfrac12 b_{xxx}+z\In b, \label{eq:rhotdeform}\\
a=-\tfrac12 b_x+\beta, \label{eq:ab}
\end{align} 
\end{subequations} 
where $\In=\rho D_x+D_x \rho$ and $\beta$ is a constant in $x$, 
which in principle depends on $t$ and also $z$.  
If, for example, $b$ is regular at $z=\infty$ and we write $\displaystyle{b=\sum_{0\leq j } \frac{b_{-j}}{ z^j}}$ then \eqref{eq:rhotdeform} reads 
\begin{subequations} 
\begin{align}
    \rho_{t}&=\In b_0, \label{eq:CGSalg}\\
    0&=\frac12 b_{-j,xxx}+\In b_{-(j+1)},  \quad 0 \leq j.  \label{eq:CGSalgcons}
\end{align} 
\end{subequations}
If we assume the shape of $b$ to be
\begin{equation}\label{eq:rationalD}
    b=b_0+\frac{b_{-1}}
    {z+\epsilon}, \qquad  
    \epsilon>0, 
\end{equation} 
then the deformation equation \eqref{eq:rhotdeform} reads
\begin{subequations} 
\begin{align}
    \rho_{t}&=\In b_0, \label{eq:CGSratt}\\
    b_{0,xxx}+\frac{b_{-1,xxx}}{\epsilon}=0, &\qquad 
    \tfrac{1}{2}b_{-1,xxx}-\epsilon \In b_{-1}=0.  \label{eq:CGSratcons}
\end{align} 
\end{subequations}

We note that equations \eqref{eq:CGSalgcons} and \eqref{eq:CGSratcons} do not involve 
any derivatives with respect to $t$ and thus can be viewed 
as representing \textit{ constraints}.  Resolving these constraints, that is 
finding a manageable parametrization of $b_j$ in terms of 
$\rho$, is  one of the essential intermediate steps in solving 
\eqref{eq:CGSalg}(or \eqref{eq:CGSratt}).  
This problem was completely solved in \cite{colville-gomez-sz} for the rational model given 
by \eqref{eq:rationalD}.

\begin{theorem}[\cite{colville-gomez-sz}] 
Let $G_{\epsilon}(x,y)$ be the Green's function 
for the boundary value problem 
\begin{equation} \label{eq:GBVP} 
D_x^2 f=\epsilon \rho f, \qquad  f'(0)-hf(0)=0, \qquad f'(1)+Hf(1)=0, \quad h\geq 0, H\geq 0
\end{equation} 
Suppose that $(h, H)\neq (0,0)$ then $b_{-1}$ and $b_0$ defined by setting 
\begin{equation}
b_{-1}(x)=G_{\epsilon}(x,x), \qquad b_{0}=\frac{\big[G_{\epsilon=0}(x,x)-G_{\epsilon}(x,x)\big]}{\epsilon}
\end{equation}
satisfy \eqref{eq:CGSratcons} and the resulting deformation 
given by \eqref{eq:CGSratt} is isospectral.  
\end{theorem} 

\begin{remark} Observe that all $b_j$ are \textit{a priori} defined only up to a quadratic 
polynomial in $x$ (see \eqref{eq:CGSratcons}).  Moreover, to satisfy the boundary conditions, 
the quadratic polynomial has to be proportional to the diagonal part,  
obtained by setting $x=y$, of the Green's function 
$G_{\epsilon=0}(x,y)$. 
\end{remark} 

It is not difficult to see that the limit $\epsilon\rightarrow 0^+$ exists 
and one obtains the following counterpart of the previous 
theorem, again proven in its entirety in \cite{colville-gomez-sz}.  
\begin{theorem} [\cite{colville-gomez-sz} ]
Let $G_0(x,y)=G_{\epsilon=0}(x,y)$ be be the Green's function 
of the operator $D_x^2$ satisfying 
\begin{equation} \label{eq:GBVP0} 
 G_{0,x}(0,y)-hG_{0}(0,y)=0, \qquad 
 G_{0,x}(1,y)+HG_{0}(1,y)=0,  \quad h\geq 0, H\geq 0
\end{equation} 
Suppose $(h, H)\neq (0,0)\, , 0\leq j$,   and define $b_{-j}$ by setting 
\begin{equation}
b_{-j}=0, \text{ for } 1<j , \qquad b_{-1}(x)=G_{0}(x,x), \qquad b_{0}=-G_{1}(x,x)\stackrel{def}{=}-\frac{dG_\epsilon(x,x)}{d\epsilon}\big|_{\epsilon=0}.  
\end{equation}
Then the $b_j$ satisfy \eqref{eq:CGSalgcons} and the resulting deformation 
given by \eqref{eq:CGSalg} is isospectral.  
\end{theorem}

\section{Deformations with finitely many fields}

The most \textit{natural} types of flows involve only finitely many fields $b_j$;   
one can formally obtain them by truncating the infinite tower of constraints at certain level $j=k$ by requiring that 
$b_{-j}=0, k<j $.  
Thus the constraints take the form
\begin{equation} \label{eq:ktrunc} 
\begin{split}
0&=\tfrac12 b_{-j, xxx}+\In b_{-(j+1)}, \qquad 0\leq j\leq k-1, \\
0&=\tfrac12 b_{-k, xxx}, \text{ and } b_{-j}=0, \qquad k<j.  
\end{split}
\end{equation}
The other finite type, generalizing the rational case $k=1$ above, can be 
taken to be 
\begin{equation}\label{eq:krat} 
b=b_0+\frac{b_{-1}}{z+\epsilon}+\frac{b_{-2}}{(z+\epsilon)^2}+\cdots+\frac{b_{-k}}{(z+\epsilon)^k}. 
\end{equation} 

 The main objective of this paper is to show that, firstly, there exists a 
 natural parametrization of the case \eqref{eq:krat} in terms of 
 the same Green's function $G_{\epsilon}$ used for $k=1$ and, secondly, that this 
 parametrization has a nontrivial limit, parametrizing the truncated case 
 \eqref{eq:ktrunc}.  
 
 To begin with, by direct computation, we get the evolution equation and the constraints for $b$ given by \eqref{eq:krat} to be
 \begin{equation}\label{eq:kratcons} 
 \begin{split} 
 \rho_t&=\In b_0, \\
 0&=\tfrac12 b_{0,xxx}+\In b_{-1}, \\
 0&=\tfrac12b_{-j, xxx} +\In b_{-(j+1)}-\epsilon \In b_{-j}, \qquad 1\leq j\leq k-1\\
 0&=\tfrac12b_{-k, xxx}-\epsilon \In b_{-k}. 
 \end{split} 
 \end{equation} 
 We will now argue that there exists a parametrization 
 of these equations in terms of the Green's function $G_{\epsilon}$ and 
 its $\epsilon$ derivatives $G_{\epsilon}^{(j)}\stackrel{def}{=} \frac{d^j G_{\epsilon}}{d \epsilon^j}$.  The proposed parametrization is given 
 by the formulas
  \begin{equation}\label{eq:kratconspar} 
 \begin{split} 
 b_0&=\frac{\big[G_{0}(x,x)-\sum_{j=0}^{k-1} \frac{G_{\epsilon}^{(j)}(x,x)}{j!}(-\epsilon)^j\big]}{\epsilon^k},  \\
 b_{-j}&=\frac{(-1)^{(k-j)} G_{\epsilon}^{(k-j)}(x,x)}{(k-j)!}, \qquad \qquad 1\leq j\leq k. 
 \end{split} 
 \end{equation} 
The main ingredient of the proof rests on the observation 
that the $\epsilon$ derivative of the last equation in \eqref{eq:kratcons} 
is, up to a correct choice of the sign, the previous equation on the list.  Indeed, differentiating once we obtain
\begin{equation} 
0=\tfrac12b_{-k, xxx}^{(1)}-\In b_{-k}-\epsilon \In b_{-k}^{(1)}, 
 \end{equation} 
and by iterating we arrive at the general formula \eqref{eq:kratconspar}.  
The formula for $b_0$ follows by successive elimination of 
terms $\In b_j$ in terms of third spacial derivatives.  The intermediate 
formula 
\begin{equation*} 
\In b_{-(k-j)}=\tfrac12 \frac{\big[b_{-k}+\epsilon b_{-(k-1)}+\cdots+ \epsilon ^j 
b_{-(k-j)}\big]_{xxx}}{\epsilon^{j+1}}
\end{equation*}
is then substituted into the first equation of \eqref{eq:kratcons} with an 
important proviso that the term $G_{0}$, which we recall is 
at most quadratic in $x$, is added to ensure the 
existence of the limit $\epsilon \rightarrow 0^+$.  Finally, the reason why 
the limit exists is that $G_{\epsilon}$ is an analytic function of $\epsilon$ 
for $\epsilon$ small enough (thanks to the assumption $(h,H)\neq (0,0)$), while the limit for $b_0$ is 
ensured by observing that the numerator in $\frac{\sum_{j=0}^{k-1} \frac{G_{\epsilon}^{(j)}(x,x)}{j!}(-\epsilon)^j}{\epsilon^k}$ is the Taylor expansion 
of $G_{0=\epsilon -\epsilon}$ about $\epsilon$ and thus it is 
equal to $G_{0}$ with an error term $\mathcal{O}(\epsilon^k)$.  
 In the limit $\epsilon \rightarrow 0^+$ 
\begin{equation*} 
b_0=(-1)^k \frac{G_{\epsilon}^{(k)}}{k!}\big|_{\epsilon=0}.  
\end{equation*}
Before we discuss the formulas for $\epsilon=0$ 
we will simplify our notation to facilitate the display of formulas.  
We will write 
\begin{equation}\label{eq:Giter}
G_{\epsilon}=G_0+G_1 \epsilon+G_2\epsilon^2+\cdots
\end{equation}
Since $G_{\epsilon}$ is the Green's function of $D_x^2 -\epsilon \rho$ 
the terms $G_j$ satisfy 
\begin{equation} 
D_x^2 G_0(x,y)=\delta (x-y), \quad D_x^2 G_{j+1}(x,y)=\rho(x) G_j(x,y), \quad 0\leq j, 
\end{equation} 
all $G_j$ subject to the boundary conditions $G_{j, x}(x=0,y)-h G_j(x=0, y)=0$, \newline $ 
~G_{j, x}(x=1,y)~+H G_j(x=1, y)=0$.  
This allows one to write an explicit formula for 
$G_j(x,y)$, 
namely, 
\begin{equation}\label{eq:Gj}
G_j(x,y)=\int \limits_{[0,1]^j} G_0(x,\xi_j)\rho(\xi_j)G_0(\xi_j, \xi_{j-1})\rho(\xi_{j-1})\cdots     \rho(\xi_1)   G_0(\xi_1,y)d\xi_j\cdots d\xi_1
\end{equation}

\begin{remark}\label{rem:Gpositivity} 
It is easy to check that $G_0(x,y)<0$ on $[0,1]$.  Since in \eqref{eq:Gj} there are 
$j+1$ factors involving $G_0$ and the remaining factors are positive we get that 
$(-1)^{j+1} G_j(x,y)>0$.  
\end{remark} 
With this notation in place we conclude that the parametrization 
\begin{equation} \label{eq:bjGpara}
b_{-j}(x)=(-1)^{k-j} G_{k-j}(x,x), \qquad 0\leq j\leq k,  
\end{equation} 
resolves the constraints \eqref{eq:ktrunc}, and in addition the $b_{-j}$s so 
defined satisfy the correct boundary conditions ensuring isospectrality \cite{colville-gomez-sz}.  
These formulae have a natural diagrammatic representation 
as illustrated in \autoref{fig:Fdiagram}, where we 
present a natural interpretation of $b_0(x)$ for the case $k=2$.  

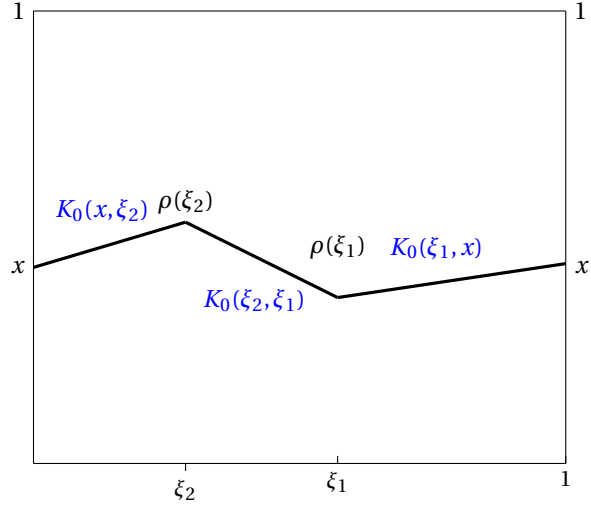
\begin{figure}[ht]
  \centering
  \begin{tikzpicture}
    \draw (0, 0) -- (7, 0);
    \draw (0, 0) -- (0, 6) ;
    \draw (7, 0) -- (7, 6) ;
    \draw (0, 6) -- (7, 6); 
    \draw (0,2.6) node[left]{$x$};
    \draw (7,2.6) node[right]{$x$};
    \draw(0,6) node[left]{$1$};
    \draw(7,6) node[right]{$1$};
    \draw(0.9,3.1) node[above]{$\color{blue} K_0(x,\xi_2)$};
    \draw(2.9,1.9) node[above]{$\color{blue} K_0(\xi_2,\xi_1)$};
    \draw(5.3,2.6) node[above]{$\color{blue} K_0(\xi_1,x)$};
    \draw(2,3.2) node[above]{$\rho(\xi_2)$};
     \draw(4,2.6) node[above]{$\rho(\xi_1)$};
    
    \def\uformula{plot (\noexpand\x,{0.3*\noexpand\x+2.6})}
    
    \def\vformula{plot (\noexpand\x,{-.5*(\noexpand\x-2)+3.2})}
    
    \def\wformula{plot (\noexpand\x,{0.15*(\noexpand\x-4)+2.2})}

    \draw[black]  (2, 0) -- +(0, -0.1) node[below] {\small $\xi_2$};
   
    \draw[black] (4, 0.1) -- +(0, -0.1) node[below] {\small $\xi_1$};
   
    \draw[black] (7, 0.1) -- +(0, -0.1) node[below] {\small $1$};
 \begin{scope}[black,very thick]
      \draw [domain = 0 : 2] \uformula;
      \draw [domain = 2 : 4] \vformula;
     \draw [domain = 4 : 7] \wformula;
    \end{scope}
  \end{tikzpicture}
  \caption{A diagram
    illustrating the flow generated by $b_0(x)$ in the case $k=2$ with $K(x,y)=-G_0(x,y)$.
    The flow is generated by a formal {\bf sum of weights along all broken 
    lines} of this type. Each edge $(\xi_{i}, \xi_{i-1})$ has a positive weight 
    ${\color{blue} K_0(\xi_i, \xi_{i-1})}$, each vertex $\xi_i$ has a positive weight $\rho(\xi_i)$ and the weight of the path is the product of weights.    
    When $\rho$ is a discrete measure (a discrete string) the formal 
    sum is a finite sum over all admissible broken lines.  }
  \label{fig:Fdiagram}
\end{figure}

We are especially interested in \textit{discrete strings}, also 
called \textit{degenerate strings} in \cite{Sabatier-NEEDS}, 
which are finite discrete measures, i.e.  $\rho =\sum_{j=1}^n m_j \delta_{x_j}, 
\, 0<x_1<x_2<\cdots<x_n<1$.  The above considerations 
carry over to this case with one important exception; as explained 
in \cite{colville-gomez-sz} 
in this case the action of $\In$ on continuous, piecewise smooth, functions $f$ 
is given by 
\begin{equation}\label{eq:Indistrib}
\In f=D_x(f(x)\rho)+\avg{f_x}(x)\rho, 
\end{equation} 
where now $D_x$ means the distributional derivative, $ \avg{f_x}(x)\delta_{x_j}=\avg{f_x}(x_j)\delta_{x_j}$, and $\avg{h}(x_j)$ denotes the 
arithmetic average of $h$ at $x_j$.  The time derivative of 
$\rho$ is easily computed to be 
\begin{equation} \label{eq:trho}
\dot \rho=\sum_{j=1}\big(\dot m_j \delta_{x_j}- \dot x_j m_j \delta_j^{(1)} \big)
\end{equation} 
and the evolution equation \eqref{eq:CGSalg} becomes a system of 
ODEs which for the flows truncated at level $k$ with $b_{-j}$ parametrized  
by \eqref{eq:bjGpara} takes the following simple form
\begin{subequations} \label{eq:xmODEs}
\begin{align} 
\dot x_j&=(-1)^{k+1} G_k(x_j, x_j), \label{eq:xdot} \\
\dot m_j&=(-1)^km_j \avg{G_{k,x}(x,x)}(x=x_j) \label{eq:mdot}. 
\end{align}
\end{subequations} 
\begin{theorem} \label{thm:HamODEs}
Equations \eqref{eq:xmODEs} are Hamiltonian with respect to the 
canonical Poisson bracket 
\begin{equation}
\{x_i, x_j\}=\{m_i, m_j\}=0, \qquad \{x_i, m_j\}=\delta_{i,j}
\end{equation} 
and the Hamiltonian 
\begin{equation} \label{eq:H} 
H^{(k)}=\frac{(-1)^{k+1}}{k+1}\int_0^1G_k(x,x) \rho(x) d\, x
=\frac{(-1)^{k+1}}{k+1}\sum_{i=1}^n m_i G_k(x_i,x_i). 
\end{equation} 
\end{theorem} 
\begin{proof} 
The computation of the $\{x_j, H^{(k)}\}$ bracket is straightforward:
\begin{equation*} 
\begin{split} 
&\{x_j, H^{(k)}\}=\frac{\partial H^{(k)}}{\partial m_j}=\\&\frac{(-1)^{k+1}}{k+1}\frac{\partial}
{\partial m_j}
 \int_{[0,1]^{k+1}} 
 \rho(\xi_{k+1})G_0(\xi_{k+1},\xi_k)\rho(\xi_k)\cdots \rho(\xi_1) G_0(\xi_1,\xi_{k+1}) d\xi _{k+1} d\xi_k\cdots d\xi_1 =\\
&\frac{(-1)^{k+1}}{k+1}\big((k+1) G_k(x_j, x_j) \big)=(-1)^{k+1} G_k(x_j,x_j), 
 \end{split} 
\end{equation*} 
where we used $(k+1)$ times that $\frac{\partial \rho }{\partial m_j}=
\delta_{x_j} $.  This proves \eqref{eq:xdot}.  
The proof of \eqref{eq:mdot} can be broken into two steps.  
First we observe that, based on \eqref{eq:Gj},
\begin{equation} \label{eq:Gkder} 
\begin{split}
&\avg{G_{k,x} (x,x)}(x=x_j)=\\&\int \limits_{[0,1]^j} \avg{G_{0,x}(x,\xi_j)}(x=x_j)\rho(\xi_j)G_0(\xi_j, \xi_{j-1})\rho(\xi_{j-1})\cdots     \rho(\xi_1)   G_0(\xi_1,y)d\xi_j\cdots d\xi_1+\\
&\int \limits_{[0,1]^j} G_0(x,\xi_j)\rho(\xi_j)G_0(\xi_j, \xi_{j-1})\rho(\xi_{j-1})\cdots     \rho(\xi_1)   \avg{G_{0, x}(\xi_1,x)}(x=x_j)d\xi_j\cdots d\xi_1.  
\end{split} 
\end{equation} 
The second observation hinges on the fact that, since $x_i\neq x_j, \, i\neq j$, 
we have that 
\begin{equation}\label{eq:Gave}
G_{0,x_i} (x_i, x_j)=\begin{cases} \avg{G_{0,x}(x,x_j)}(x=x_i), \quad i\neq j\\
                                                       2 \avg{G_{0,x}(x,x_i)}(x=x_i), \quad i=j, 
                                 \end{cases} 
\end{equation}
which follows immediately from the shape of the Green's function, 
namely, 
\begin{equation*} 
G_0(x_i, x_j)=\begin{cases} c(x_i)\hat c(x_j), \quad  i<j\\ 
                                              c(x_i)\hat c(x_i),\quad  i=j\\
                                              c(x_j)\hat c(x_i) , \quad j<i
                                              \end{cases} 
\end{equation*} 
where $c(x), \hat c(x)$ are linear functions of $x$, satisfying appropriate 
boundary conditions at $x=0$, $x=1$, respectively.       
We need to compute, for simplicity expanding integrals in terms of sums, 
\begin{equation*} 
\begin{split} 
&\{m_j, H^{(k)} \}=-\frac{\partial H^{(k)} }{\partial x_j}=\\&\frac{(-1)^{k}}{k+1}\frac{\partial}
{\partial x_j}
\,  \sum_{i_{k+1}, i_{k}, \cdots i_1}
 m_{i_{k+1}} m_{i_{k}}\cdots m_{i_1}G_0(x_{i_{k+1}},x_{i_{k}})G_0(x_{i_{k}},x_{i_{k-1}})\cdots G_0(x_{i_{1}},x_{i_{k+1}})  =\\
& (-1)^{k}
 \sum_{i_{k+1}, i_{k}, \cdots, i_1}
 m_{i_{k+1}} m_{i_{k}}\cdots m_{i_1}\frac{\partial}
{\partial x_j}\big(G_0(x_{i_{k+1}},x_{i_{k}})\big) G_0(x_{i_{k}},x_{i_{k-1}})\cdots G_0(x_{i_{1}},x_{i_{k+1}}), 
 \end{split} 
\end{equation*}   
where in the second step we used that the expression under the sum 
is symmetric. Moreover, 
\begin{equation*} 
G_{0, x_j}(x_{i_{k+1}},x_{i_{k}})=\delta_{i_{k+1}, j} 
G_{0, x_j}(x_j, x_{i_k})+(1-\delta_{i_{k+1}, j} )\delta_{i_{k},j}G_{0, x_j}(x_{i_{k+1}}, x_j),  
\end{equation*} 
and after substituting back into the above formula and making use 
of \eqref{eq:Gave} we obtain 
\begin{equation*}
 \begin{split} 
&\{m_j, H^{(k)} \}= (-1)^{k}
m_j \big(\sum_{i_{k}, i_{i}, \cdots, i_1}
  m_{i_{k}}\cdots m_{i_1}
\avg{G_{0,x}(x,x_{i_{k}})}(x=x_j) G_0(x_{i_{k}},x_{i_{k-1}})\cdots G_0(x_{i_{1}},x_j)+\\
&\sum_{i_{k}, i_{i}, \cdots, i_1}
  m_{i_{k}}\cdots m_{i_1}
G_0(x_{i_{k}},x_{i_{k-1}})\cdots \avg{G_{0, y}(x_{i_{1}},y)}(y=x_j)\big)\stackrel{\eqref{eq:Gkder}}{=}\\
&(-1)^km_j \avg{G_{k,x}(x,x)}(x=x_j),
\end{split} 
\end{equation*} 
which  confirms \eqref{eq:mdot}.                                       
\end{proof} 
\begin{remark} 
In view of \eqref{rem:Gpositivity} we see that Hamiltonians 
given by \eqref{eq:H} are non-negative.  Moreover, 
they can be shown to be simply related to 
spectral invariants (see i.e. \cite{colville-gomez-sz})
\begin{subequations} \label{eq:Invariants} 
\begin{align}
&I_0=\int_{[0,1]} \rho(\xi_1)G_0(\xi,\xi) \, d\xi\\
&I_j=\int_{0<\xi_1<\xi_2<\cdots<\xi_j<1}\rho(\xi_j)(\xi_j-\xi_{j-1})\cdots
\rho(\xi_2)(\xi_2-\xi_1)\rho(\xi_1) G_0(\xi_1, \xi_j) \, d\xi\cdots d\xi_j. 
\end{align} 
\end{subequations}
For example, if $k=1$ then for $H^{(1)} $ computed from \eqref{eq:H} we get 
\begin{equation*} 
H^{(1)}=\tfrac12 I_1^2 + I_2. 
\end{equation*} 
\end{remark} 
\subsection{Evolution of spectral data} 
We will briefly discuss how one can solve equations \eqref{eq:xdot}, 
\eqref{eq:mdot}, or alternatively $\rho_t=\In b_0$ 
with $b_{0}(x)=(-1)^k G_k(x,x)$  
and 
\begin{equation} \label{eq:rho} 
\rho=\sum_{j=1}^nm_j\delta_{x_j}, 
\qquad 0<x_1<x_2<\cdots<x_n<1. 
\end{equation}   
Let $\phi$ satisfy \eqref{eq:string} and the left initial condition $\phi_x(0)-h\phi(0)=0$.  Then for $0\leq x<x_1$ (the position of the first mass) we 
can take 
\begin{equation}\label{eq:phi1}
\phi(x)=hx+1.
\end{equation} 
Since $\phi$ changes with the deformation time according to \eqref{eq:vtdeform} we have that for $0<x<x_1$ 
\begin{equation}
0=(-\tfrac12 b_x+\beta)(hx+1)+h b 
\end{equation} 
which implies 
\begin{equation}
\beta=\tfrac12 b_x(0)-hb(0). 
\end{equation} 
In fact, for $b=b_{0}+\frac{b_{-1}}{z}+\cdots+ \frac{b_{-k}}{z^k}$ and $b_{-j} $ given by \eqref{eq:bjGpara}, 
$\beta$ takes the following simple form 
\begin{equation}\label{eq:betaz}  
\beta(z)=\frac{1}{2z^k}.  
\end{equation} 
The spectrum is given by the zeros of the function 
$D(z)=\phi_x(1;z)+H\phi(1;z)$ for $0\leq H<\infty$ and $D(z)=\phi(1;z)$ 
for $H=\infty$, and we have 
the following linearization result
\begin{theorem} \label{thm:linearization}
Let \, $0<H<\infty$,  and let 
\begin{equation} 
N(z;t) \stackrel{def}{=}\phi_x(1;z)-H\phi(1;z).  
\end{equation}
Then for every point $z_i$ of the spectrum the time derivative of $N(z_i;t)$ 
satisfies
\begin{equation}
\dot N(z_i;t)=\tfrac{1}{z_i^k }N(z_i;t). 
\end{equation} 
\end{theorem} 
\begin{proof} 
We start by computing 
$\dot N(z_i;t)=\phi_{t,x}(1;z_i)-H\phi_t(1;z_i)$ 
with the help of \eqref{eq:vtdeform}, \eqref{eq:ab} and 
employing one intermediate result proven in \cite{colville-gomez-sz} 
stating that one of the necessary conditions for isospectrality can be written 
\begin{equation*}
\tfrac12 b_{xx}(1)+Hb_x(1)+H^2 b(1)=0. 
\end{equation*}
Then, by straightforward computation, and the fact that on the 
spectrum $D=0$, we obtain: 
\begin{equation*}
\dot N(z_i;t)=\big(\beta(z_i)-\frac{b_x(1;z_i)+2H b(1;z_i)}{2} \big) N(z_i;t). 
\end{equation*}
For the flows $b=b_{0}+\frac{b_{-1}}{z}+\cdots+ \frac{b_{-k}}{z^k}$ and 
$b_{-j} $ given by \eqref{eq:bjGpara} the second term is zero, when 
evaluated at $x=1$, except 
for the last term $\frac{b_{-k}}{z_i^k} $ for which one obtains $\frac{-1}{z_i^k}$.  The proof of this 
claim follows from 
$G_{0,x}(x;y)(x=1, y\neq 1)+HG_0(x;y)(x=1, y\neq 1)=0$ and 
$G_{0,x}(x;x)(x=1)+HG_0(x;x)(x=1)=-1$.  Finally, using \eqref{eq:betaz}, we obtain the required result.  
\end{proof} 
It is now easy to cover the remaining two special cases $H=0, \, H=\infty$.  
We modify the definition of $N$; we set $N=-\phi(1;z)$ and $N(z;t)=\phi_x(1;z)$ respectively.  
\begin{theorem} \label{thm:linearizationbis}
Let $H=\infty$ or $H=0, (h,H)\neq 0$.    
Then for every point $z_i$ of the spectrum 
\begin{equation}
\dot N(z_i;t)=\tfrac{1}{z_i^k }N(z_i;t).  
\end{equation} 
\end{theorem} 
\begin{proof} 
We will give the proof in the case $H=\infty$, leaving the case 
$H=0$ for interested readers.  
Following the same steps as above we obtain: 
\begin{equation*}
\dot N(z_i;t)=\big(\beta(z_i)+\tfrac12 b_x(1)\big) N(z_i;t), 
\end{equation*}
and subsequently observe that the only change now is that $G_0(x;y)(x=1, y\neq 1)=0$ and 
$G_{0,x}(x;x)(x=1)=-1$, which in conjuncture with \eqref{eq:betaz} imply the required result.  
\end{proof} 
\begin{remark} 
The reason why the case $(h,H)=(0,0)$ is excluded is 
because the formula for the Green's function using iterations 
\eqref{eq:Giter} does not hold since $z=0$ is now in the spectrum.  
\end{remark} 
As is well known the Weyl function (Weyl-Titchmarch function) 
is a convenient way of storing information about the boundary 
value problem \eqref{eq:string}.  We can define it
as 
\begin{equation}\label{eq:W}
W(z)=\frac{N(z)}{D(z)}.  
\end{equation}
$W(-z)$ is an example of a \textit{Herglotz} function (i.e. p.17 in \cite{Deift-book}), in particular $-W(-z)$ admits an integral Stieltjes
representation
\begin{equation}\label{eq:Stieltjes rep}
-W(-z)=\gamma+\int \frac {d\mu(\zeta)}{z+\zeta}, \qquad \gamma \in \RR, 
\end{equation}
where $d\mu$ is a positive measure supported on positive reals $\RR_+$.  
In the case of interest, when dealing with discrete strings, $d\mu=\sum_{j=1}^n
\mu_j \delta_{z_j}, \, \mu_j>0$, where $z_j$ are the eigenvalues 
of the boundary value problem.   The dependence on the deformation 
parameter, in view of \eqref{thm:linearization} and \ref{thm:linearizationbis}, is given by $\mu_j(t)=\mu_j(0) e^{\tfrac{t}{z_j^k}}$ for the flow 
given by \eqref{eq:ktrunc}.  
\begin{example} \label{ex:first}
We will illustrate the developed formalism on  the case of the string boundary value problem \eqref{eq:string} corresponding to the mass density \eqref{eq:rho}, 
with boundary conditions  of the type $0<h<\infty$ and $H=0$. 
The (bare) Green's function in this case is 
\begin{equation*}
G_0(x,y)= \begin{cases} \tfrac{hx+1}{-h}, \quad x<y\\
                                    \tfrac{hy+1}{-h}, \quad y<x. 
                 \end{cases}
\end{equation*}
With the help of \eqref{eq:Gj}, \eqref{eq:bjGpara} we can 
now write all coefficients in \eqref{eq:ktrunc}, in particular 
obtaining 
\begin{equation}
b_0(x)=(-1)^k \sum_{j_1, j_2, \cdots, j_k} G_0(x, x_{j_k})m_{j_k} 
G_0(x_{j_k}, x_{j_{k-1}}) m_{j_{k-1}}\cdots m_{j_1} G_0(x_{j_1},x). 
\end{equation}
We observe that the first equation of motion, \eqref{eq:xdot}, 
can be written in a compact way if one defines 
matrices $M=\textrm{diag}(m_1, \dots, m_n)$ and $K=[K_{i, j}]=[-G_0(x_i, x_j)]$.  
Then 
\begin{equation} \label{eq:dotxmatrix}
\dot x_j=(KMKM...MK)_{j, j}.  
\end{equation} 

In particular, if $k=1$, we 
get $\dot x_j=(KMK)_{j, j}$.  Likewise, the second equation \eqref{eq:mdot} 
can be written in an analogous way by introducing 
a matrix $J=[J_{i, j}]=2\avg{G_{0,x}(x, x_j)(x=x_i)}$, resulting in 
\begin{equation} \label{eq:dotmmatrix}
\dot m_j=(MJMKMK\cdots MK)_{j,j}.  
\end{equation}

We now turn to 
outlining the integration process for the mass density \eqref{eq:rho}.

Let
\begin{equation*}
	\phi |_{I_{j}}=\phi_{j}=p_{j}(x-x_{j})+q_{j},\qquad \text{where } I_{j}=(x_{j},x_{j+1}), \text{ and } x_0=1, \, x_{n+1}=1, 
\end{equation*}
denote the solution to the initial value problem $-\phi_{xx}=z\rho \phi, \,\,   \phi_x(0)-h\phi(0)=0$ whose construction proceeds as follows.  
We start off with $p_{0}=h$ and $q_{0}=1$ to satisfy the initial condition 
at $x=0$.  Then letting $l_{j}$ denote the length of the interval $I_j$ and imposing the continuity of $\phi$ at $x=x_{j+1}$ one obtains that $p_{j}=\frac{q_{j+1}-q_{j}}{l_{j}}$.  The jump in the derivative of $\phi$ at $x=x_{j+1}$ results in $p_{j+1}-p_{j}=-zm_{j+1}q_{j+1}$. On the last interval $I_n$, $\phi_{n}=p_{n}(x-x_{n})+q_{n}$.  Let us define the Weyl function for this problem
\begin{equation*}
	W(z)\stackrel{\emph{def}}{=}-\frac{\phi(1;z)}{\phi_{x}(1;z)}. 
\end{equation*}
From the construction $-W(-z)=\frac{p_n(-z)l_n+q_n(-z)}{p_n(-z)}=l_n +\frac{q_n(-z)}{p_n(-z)}$. 
Iterating with the help of continuity and jump conditions we obtain 
\begin{equation*}
-W(-z)=l_{n}+\cfrac{1}{zm_n+\cfrac{1}{l_{n-1}+\cfrac{1}{\ddots+\cfrac{1}{l_{0}+\frac{1}{h}}}}},  
\end{equation*}
and, upon comparing with \eqref{eq:Stieltjes rep}, 
we obtain 
\begin{equation}
\gamma=l_n, \qquad 
\int \frac{e^{\tfrac{t}{\zeta^k}} d\mu(\zeta;0)}{z+\zeta}=\cfrac{1}{z m_n+\cfrac{1}{l_{n-1}+\cfrac{1}{\ddots+\cfrac{1}{l_{0}+\frac{1}{h}}}}}, 
\end{equation}
where we denote $d\mu(\zeta;0)=d\mu(\zeta;t=0)$, 
which shows that using Stieltjes' inversion formulas (see \cite{stieltjes, BSS-Stieltjes}) we can recover $m_n, \cdots, m_1$ and $l_{n-1}, \cdots, l_0$ in 
terms of the Hankel determinants of the moments of the measure $e^{\tfrac{t}{\zeta^k}} d\mu(\zeta;0)$.  
We will briefly review the relevant part of Stieltjes' theory.  
In \cite{stieltjes} T. Stieltjes studied the 
continued fraction 
\begin{equation}\label{eq:contdfrac}
f(z)=\cfrac{1}{a_1z +\cfrac{1}{a_2+\cfrac{1}{\qquad  \ddots +\cfrac{1}{a_{2j-1}z+\frac{1}{a_{2j}+ \dots
\ddots }}}}}
\end{equation}
under the assumption $a_j>0$, viewed as a function of the 
complex variable $z$.  
Then he considered 
a formal Laurent expansion at $z=\infty$ of the continued fraction, written 
as 
\begin{equation*}
\frac{c_0}{z}-\frac{c_1}{z^2}+\frac{c_3}{z^3}+\cdots
\end{equation*}
Then the main thrust of Stieltjes's theory went towards establishing 
the existence of a measure (Stieltjes measure) $d\alpha$ supported 
on $\RR_+=[0,\infty)$ such that 
$c_j=\int_{\RR_+} \zeta^j d\alpha(\zeta)$.  
In the case of interest this measure is unique (because the Laurent series 
converges) and is simply 
given by 
\begin{equation} \label{eq:measure} 
d\alpha(\zeta)=e^{\tfrac{t}{\zeta^k}} d\mu(\zeta;0)=\sum_{j=1}^n 
\mu_j(0)e^{\tfrac{t}{z_j^k}} \delta_{z_j}.
\end{equation} 
Then one defines the Hankel matrix of moments
\begin{equation}\label{eq:Hankel} 
H=\begin{pmatrix} c_{-1}&c_0&c_1&\cdots& c_j&\cdots\\
                             c_{0}&c_1&c_2&\cdots& c_{j+1}&\cdots\\
                             c_{1}&c_2&c_3&\cdots& c_{j+2}&\cdots\\
                             \cdots&\cdots&\cdots&\cdots&\cdots\\
                             c_{l-1}&\cdots&\cdots&\cdots& c_{j+l}&\cdots
                             \end{pmatrix}, 
\end{equation} 
  where we included $c_{-1}$ because in our case 
  the measure $d\alpha$ is supported away from $0$.  With the 
  help  of $H$ one introduces certain minors, called 
  Hankel determinants, denoted $\Delta_k^l$, which are defined 
  as the determinants of $k\times k$ submatrices whose $(i, j)$ entries
  are $c_{i+j+l-2}$, while $\Delta_0^l=1$ by convention.  
    Finally, one can express the coefficients $a_j$ in \eqref{eq:contdfrac}
  in terms of these determinants using the formulae
    \begin{equation}\label{eq:Sinversion}
a_{2j}=\frac{(\Delta_j^0)^2}{\Delta_j^1 \Delta_{j-1}^1}, \qquad a_{2j+1}=\frac{(\Delta_j^1)^2}{\Delta_j^0 \Delta_{j+1}^2}. 
\end{equation}
For simplicity let us set  $j'=n-j$, then an immediate application of these formulae to our case yields: 
\begin{subequations}
\begin{align} 
m_j&=\frac{(\Delta_{j'}^1)^2}{\Delta_{j'}^0 \Delta_{(j-1)'}^2}, \qquad  1\leq j\leq n\label{eq:mj}\\
l_j&=\frac{(\Delta_{j'}^0)^2}{\Delta_{j'}^1 \Delta_{(j+1)'}^1}, \qquad 1\leq j\leq  n-1\label{eq:lj}\\
l_0+\tfrac 1h&=\frac{(\Delta_{0'}^0)^2}{\Delta_{0'}^1 \Delta_{1'}^1}, \qquad 
\label{eq:l0}
\end{align} 
\end{subequations}
                             
In the final step,  one can recover $l_n$ by observing that 
\begin{equation*}
-W(0)=\frac{h+1}{h}=1+\frac1h=l_n+\int \frac{e^{\frac{t}{\zeta^k}}}{\zeta} d\mu(\zeta;0)=l_n+\Delta_{1}^{-1}
\end{equation*}
and solve for $l_n$, or compute $l_n$ from the formula for the 
total length of the string: $l_n=1-\sum_{i=0}^{n-1}l_i$.  
We conclude this example by noting that the case $h=\infty$, which 
means the Dirichlet condition on the left end and the Neumann condition 
on the right, can be handled by taking the limit $h\rightarrow 0$.  Interestingly, the Dirichlet-Neumann case appeared, 
somewhat unexpectedly, in the recent work on the modified Camassa-Holm equation
\cite{ChangS}.  
\end{example} 

\begin{example} \label{ex:CH}
The case of the CH equation \eqref{eq:CH} is not the main 
focus of this paper.  However, one might be tempted to compare the 
known formulas with what one gets if the formalism is applied to that 
case.  The literature on the CH equation is so vast that one can 
not do justice to many important contributions to the subject.  We will 
only refer to papers the results of which overlap in scope 
with the presented formalism.  One way of looking at 
the CH theory is to start from 
the spectral problem 
\begin{equation}\label{eq:CHx}
-v_{xx}+\tfrac14 v=zm v,    \qquad -\infty <x<\infty, 
\end{equation}
with vanishing boundary conditions $v\rightarrow 0$ as $\abs{x}\rightarrow \infty$ \cite{CH}.  We note that this spectral problem 
also appears in \cite{Sabatier-Constants}.  In \eqref{eq:CHx} $m$ 
is a measure, similar to the string mass density $\rho$ in 
\eqref{eq:string}.  For convenience we will 
assume that $m$ has a compact support as this is sufficient for our 
purposes.  We deform $m$ in exactly the same way as in 
\eqref{eq:vtdeform}, i.e. 
\begin{equation*}
v_t=av+bv_x
\end{equation*}
and from the condition $v_{txx}=v_{xxt}$ we 
get 
\begin{subequations}
\begin{align}
zm_t&=\tfrac12 b_{xxx}-\tfrac12 b_x +z \Inm b, \\
a&=-\tfrac12 b_x +\beta. 
\end{align}
\end{subequations}
For $b$ regular at $z=\infty$ the evolution 
equation is the same as for the string, namely, 
\begin{equation*}
m_t=\Inm b_0, 
\end{equation*} 
where $b_0$ is the limit of $b$ at $z=\infty$.  However, 
the constraints are 
different.  Let us analyze the constraints for the rational model 
specified in \eqref{eq:krat}.  The resulting constraints 
can be simply obtained by changing $D_x^3$ to $D_x^3-D_x$ 
in \eqref{eq:kratcons}.  
Let us denote by $G_{\epsilon}(x,y)$, the 
Green's function for $D_x^2-\tfrac14-\epsilon m$ vanishing as $\abs{x}\rightarrow \infty$.  For these boundary conditions
\begin{equation*} 
G_{\epsilon}=G_0+\epsilon G_1+\epsilon^2 G_2+\cdots 
\end{equation*}
where 
\begin{equation} 
G_0(x,y)=-e^{\tfrac{\abs{x-y}}{2}}. 
\end{equation} 
The perturbative expansion produces essentially the same formula 
as in \eqref{eq:Gj}, namely 
\begin{equation}\label{eq:Gjbis}
G_j(x,y)=\int \limits_{\RR ^j} G_0(x,\xi_j)m(\xi_j)G_0(\xi_j, \xi_{j-1})m(\xi_{j-1})\cdots     m(\xi_1)   G_0(\xi_1,y)d\xi_j\cdots d\xi_1.  
\end{equation}
\end{example}
Now we only check that  the diagonal of the Green's function, $G_{\epsilon}(x,x)$, satisfies 
\begin{equation}
G_{\epsilon, xxx}(x,x)-G_{\epsilon, x}(x,x)=\epsilon \Inm G_{\epsilon}(x,x), 
\end{equation}
the proof of which is essentially identical to the one in Lemma 5.1 in \cite{colville-gomez-sz}.  This result, in conjuncture with 
the fact that, after changing $D_x^3$ to $D_x^3 -D_x$, 
the last equation in \eqref{eq:kratcons} is 
satisfied by $G_{\epsilon}(x,x)$ and the iterative procedure employed in  
the analysis of rational flows for the string goes through, resulting 
in the validity of the final formulae \eqref{eq:kratconspar}.  In particular, 
the limit $\epsilon \rightarrow 0^+$ can be carried out yielding an explicit 
parametrization of the flow $b=b_0+\tfrac{b_{-1}}{z}+\cdots+ \tfrac{b_{-k}}{z^k}$, structurally the same as \eqref{eq:bjGpara}, namely
$ 
b_{-j}(x)=(-1)^{k-j} G_{k-j}(x,x).  
$ 
The most relevant physically is the case of the discrete measure 
$m=\sum_{j}^n m_j \delta_{x_j}$ (the peakon sector for $k=1$ 
\cite{CH}).  The equations of motion for $x_j$ and $m_j$ are 
given by equations \eqref{eq:dotxmatrix} and \eqref{eq:dotmmatrix}
with $K=[e^{-\frac{\abs{x_i-x_ j}}{2}}]$ and $J=[\sgn(x_j-x_i)e^{-\tfrac{\abs{x_j-x_i}}{2}}]$.   For example, 
for $k=1$ we 
obtain 
\begin{align*} 
\dot x_j&=(KMK)_{j,j}=\sum_{i}m_i e^{-\abs{x_j-x_i}}, \\
\dot m_j &=(MJMK)_{j,j} =m_j \sum_i \sgn(x_j-x_i)m_i e^{-\abs{x_j-x_i}}, 
\end{align*} 
well known from the CH theory \cite{chh}.  

For $k=2$ we obtain 
\begin{align*} 
\dot x_j&=(KMKMK)_{j,j}=\sum_{i, l} e^{-\tfrac{\abs{x_j-x_i}}{2}}m_ie^{-\tfrac{\abs{x_i-x_l}}{2}}m_le^{-\tfrac{\abs{x_l-x_j}}{2}}, \\
\dot m_j &=(MJMKMK)_{j,j}=m_j \sum_{i,l} \sgn(x_j-x_i)e^{-\tfrac{\abs{x_j-x_i}}{2}}m_ie^{-\tfrac{\abs{x_i-x_l}}{2}}m_le^{-\tfrac{\abs{x_l-x_j}}{2}}, 
\end{align*} 
clearly indicating a general pattern for arbitrary $k$.  Although 
the higher order flows in the CH theory have been investigated 
\cite{McKean-Fred, schiffdual, qiao-ch} we are not aware of any 
work on the peakon sector of the CH hierarchy.  In particular 
we conclude that explicit integration of higher-order peakon flows 
will be possible using Stieltjes continued fractions generalizing 
the results of \cite{BSS-Stieltjes}.  We leave this however for future work.  
\section{Acknowledgements} 
The second author was supported in part by the Natural  Sciences and Engineering Research Council of Canada (NSERC).  He also would 
like to thank the H. Niewodnicza\'nski Institute of Nuclear Physics, Polish Academy of Sciences, 
for hospitality during the summer 
of 2017 when the work reported in this paper was carried out.  
\def\cydot{\leavevmode\raise.4ex\hbox{.}}
  \def\cydot{\leavevmode\raise.4ex\hbox{.}}

\end{document}